\documentclass[11pt]{article}
\usepackage{supertech}
\usepackage{xspace}
\newcommand{\suu}{\textsf{SUU}\xspace}

\newcommand{\suureform}{$\textsf{SUU}^*$\xspace}

\newcommand{\suui}{\textsf{SUU-I}\xspace}
\newcommand{\stochi}{\textsf{STOCH-I}\xspace}
\newcommand{\suuc}{\textsf{SUU-C}\xspace}
\newcommand{\suut}{\textsf{SUU-T}\xspace}

\newcommand{\schedname}[1]{\ifmmode{\rm #1}\else#1\fi\xspace}

\newcommand{\ALGobl}{\schedname{SUU-I-OBL}}
\newcommand{\ALGsem}{\schedname{SUU-I-SEM}}
\newcommand{\ALGstci}{\schedname{STC-I}}
\newcommand{\ALGchn}{\schedname{SUU-C}}

\newcommand{\Topt}{T_{\rm OPT}}
\newcommand{\Tsig}{T_{\Sigma}}
\newcommand{\Toff}{T_{\rm OFF}}

\newcommand{\poly}{{\rm poly}}

\def\hyperspc{\kern -0.22em}

\newcommand{\setthmcount}[1] {
  \newcounter{#1}
  \setcounter{#1}{\thetheorem}
}

\renewcommand{\set}[1]            {\{ #1 \}}
\renewcommand{\abs}[1]            {| #1|}
\renewcommand{\card}[1]           {| #1|}
\renewcommand{\floor}[1]          {\lfloor #1 \rfloor}
\renewcommand{\ceil}[1]           {\lceil #1 \rceil}
\renewcommand{\ang}[1]            {\ifmmode{\langle #1 \rangle}
                                 \else{$\langle${#1}$\rangle$}\fi}
\renewcommand{\prob}[1]           {\Pr\{ #1 \}}
\renewcommand{\expect}[1]         {{\rm E}[ #1 ]}

\newcommand{\bigo}[1]           {O\!\left(#1\right)}

\begin{document}
\title{Improved Approximations for Multiprocessor Scheduling Under Uncertainty \\
  }

\author{Christopher Crutchfield \and Zoran Dzunic 
\and Jeremy T. Fineman\thanks{Supported in part by Google, NSF Grant CSR-AES 0615215.}
\and David R. Karger 
\and Jacob H. Scott\thanks{Supported by an NDSEG Fellowship.} \\ \\
   Computer Science and Artificial Intelligence Laboratory \\
   Massachusetts Institute of Technology \\
   Cambridge, MA 02139, USA \\
   \texttt\{cyc,zoki,jfineman,karger,jhscott\}@csail.mit.edu}

\maketitle
\thispagestyle{empty}

This paper presents improved approximation algorithms for the problem
of \defn{multiprocessor scheduling under uncertainty} (\suu), in which
the execution of each job may fail probabilistically.  This problem is
motivated by the increasing use of distributed computing to handle
large, computationally intensive tasks.  In the \suu problem we are
given $n$ unit-length jobs and $m$ machines, a directed acyclic graph
$G$ of precedence constraints among jobs, and unrelated failure
probabilities $q_{ij}$ for each job~$j$ when executed on machine~$i$ for
a single timestep.  Our goal is to find a schedule that minimizes the
expected makespan, which is the expected time at which all jobs
complete.

Lin and Rajaraman gave the first approximations for this NP-hard
problem for the special cases of independent jobs, precedence
constraints forming disjoint chains, and precedence constraints
forming trees.  In this paper, we present asymptotically better
approximation algorithms.  In particular, we give an
$\bigo{\log\log(\min\set{m,n})}$-approximation for independent jobs
(improving on the previously best $\bigo{\log n}$-approximation).  We also
give an $\bigo{\log(n+m) \log\log(\min\set{m,n})}$-approximation algorithm
for precedence constraints that form disjoint chains (improving on the
previously best
$\bigo{\log(n)\log(m)\frac{\log(n+m)}{\log\log(n+m)}}$-approximation by a
$(\log n/\log\log n)^2$ factor when $n = m^{\Theta(1)}$.  Our
algorithm for precedence constraints forming chains can also be used
as a component for precedence constraints forming trees, yielding a
similar improvement over the previously best algorithms for trees.


\newpage
\setcounter{page}{1}

\secput{intro}{Introduction} 

Our work concerns approximation algorithms for multiprocessor
  scheduling under uncertainty, first introduced in
\cite{Malewicz05}.  This model extends the classical construction of
machine scheduling to handle cases where machines run jobs for
discrete timesteps and succeed in processing them only
probabilistically.  Our motivation stems from the increasing use of
distributed computing to handle large, computationally intensive
tasks. Projects like Seti@Home~\cite{Anderson02} divide computations
into smaller jobs of relatively uniform length, which are then
executed on unreliable machines (e.g., of volunteers).

Scheduling multiple machines to process the same job at once can help
overcome the problem of unreliable machines, but many machines
processing a single job can also slow down overall throughput. The
situation is exacerbated when precedence constraints among jobs are
present, which is often the case for sophisticated computations; here
a single job failing may delay the start of many others. Note that the
special case of having no precedence constraints retains practical
significance. Google's MapReduce architecture \cite{DeanG04}, for
example, generates jobs whose dependencies form a complete bipartite
graph, which is equivalent to two phases of independent jobs.
  
Motivated by these examples, we study the \defn{multiprocessor
  scheduling under uncertainty} (\suu) problem.  An \suu instance is
comprised of a set of $n$ unit-time jobs and a set of $m$ machines.
For each machine $i$ and job $j$, we are given a failure probability
$q_{ij}$, which is the chance that job $j$ does not complete when run
on machine $i$ for a single timestep.  Any precedence constraints are
modeled as a directed acyclic graph (dag).  Our objective is to
construct a schedule assigning machines to eligible jobs at each
timestep, minimizing the expected time until all jobs have
successfully completed.  In contrast to many other scheduling
problems, \suu allows multiple machines to execute the same job in a
single timestep.

\punt{When properly reformulated (in \secref{reform}), the \suu
  problem is strikingly similar to a stochastic scheduling problem,
  often referred to by
  $R|\id{pmtn},\id{prec},p_j\!\!\sim\!\!\id{stoch}|\expect{C_{\max}}$ in the
  scheduling literature~\cite{XXX}. This is the problem of
  preemptively scheduling precedence-constrained jobs whose processing
  times are given by some (known) distributions, on unrelated parallel
  machines so as to minimize their expected makespan.  Here, there are
  no failure probabilities.  Instead, each job~$j$ has an processing
  time $p_j$ that is chosen stochastically from some (known)
  distribution, and each machine has an (unrelated) speed for each
  job; typically, the distributions are exponential.  The job
  completes when the work done matches the stochastically chosen
  processing time.  The major differences in \suu are that multiple
  machines may be assigned to a job at the same time, jobs must be
  scheduled at a unit granularity, and all jobs have the same
  distribution over processing times. } 

\subheading{Related work}

Malewicz's initial presentation of \suu~\cite{Malewicz05} includes a
polynomial-time dynamic-programming solution for instances where both
the number of machines and the width of the precedence dag are
constant. If either of these constraints is relaxed, he proves that
the problem becomes NP-Hard.  Furthermore, when both constraints are
removed, there is no polynomial-time approximation algorithm for the
problem achieving an approximation ratio lower than 5/4, unless
$P=NP$.  This work does not include approximation algorithms for the
general (NP-Hard) problem.

Lin and Rajaraman present the first (and, to date, only) approximation
algorithms for \suu~\cite{LinRa07}.  Using a greedy approximation
algorithm to maximize the chance of success across all jobs, they give
an $\bigo{\log n}$-approximation when all jobs are independent. More
sophisticated techniques, including LP-rounding and random
delay~\cite{LeightonMaRa94, ShmoysStWe91}, yield a variety of
$\bigo{\poly\log(n+m)}$ approximations when precedence graphs are
constrained to form only disjoint chains, collections of in- or
out-trees, and directed forests. For these settings, our algorithms
improve their approximation ratios by a $(\log n / \log\log n)^2$
factor, when $n=m^{\Theta(1)}$. See Table \ref{table:comparison} 
for a complete comparison.

The wider field of machine scheduling is an established and
well-studied area of research with a large number of variations on its
core theme (see \cite{Karger97} for a survey). There are three main differences
between \suu and problems studied in the literature.  First, in \suu,
jobs may run on multiple machines in the same timestep.  Second, each
job has a chance of failing to complete on any machine
that processes it.  Third, jobs must be scheduled at unit granularity.

\punt{
There is a large body of work on scheduling problems involving
uncertainty, but the work in this vein commonly has jobs complete with
certainty on all machines, and instead allows the processing time
required for each one to be set according to some probability
distribution.  We show in \secref{stoch}\note{Jer:Do we show it in the
  reformulation section or the stochastic section?} that \suu is
closely related to the problem of preemptively scheduling
precedence-constrained jobs whose processing times are given by
exponential distributions, on unrelated parallel machines so as to
minimize their expected makespan.  In \defn{stochastic scheduling},
this problem is known as $R|\id{pmtn}, \id{prec}, p_j\sim\id{stoch} |
\expect{C_{\max}}$.  We know of no previous nontrivial approximation
algorithms for this problem, although $O(1)$-approximation algorithms
exist\note{Jer: Actually, aren't they optimal algorithms, not
  approximations?} when the machines are related.\note{Need citations
  here}
}

Of deterministic scheduling problems, \suu most closely resembles
$R|\id{prec},\id{pmtn}|C_{\max}$\cite{LeightonMaRa94, ShmoysStWe91,
  KumarMaPa05}, the problem of preemptively scheduling jobs with
precedence constraints on unrelated parallel machines so as to
minimize makespan.  Instead of failure probabilities $q_{ij}$, there
is a deterministic processing time $p_{ij}$, denoting how long it
takes for machine~$i$ to complete job~$j$.  In contrast to \suu,
however, machines never fail, and jobs may only run on one machine at
a time.  As in~\cite{LinRa07}, techniques for this problem play an
important role in our approximations.  We also borrow techniques from
``job-shop scheduling''~\cite{French82} for our \suu algorithms for precedence
constraints, but the particulars of that setting are not very similar
to the ones we consider here.  
\punt{
We are only aware of approximations for
$R|\id{prec},\id{pmtn}|C_{\max}$ when precedence constraints form
disjoint chains or trees.
}

There is also a large body of work in stochastic scheduling (see
\cite[Part 2]{Pinedo02} for a representative sample). The majority of
the work in this area considers how to schedule jobs whose input
lengths are not known, but instead given as random variables
distributed according to some probability distribution. Particular
attention has been paid to the case when these distributions are
restricted to exponential families, which is similar in some respects
to \suu (see \appref{stoch}).  However, we know of no approximation
algorithms,  even with this restriction, when machines are unrelated.

\punt{There is a large body of work in stochastic scheduling.  For unrelated
machines ($R|\id{pmtn},\id{prec},p_j\!\!\sim\!\!\id{stoch}|\expect{C_{\max}}$), we
know of no previous nontrivial approximation algorithms.  There are,
however, $O(1)$-approximations when machines are related (i.e., for
$P|\id{pmtn},\id{prec},p_j\sim\id{stoch}|\expect{C_{\max}}$), and all
processing times are from exponential distributions~\cite{} and other
nice distributions~\cite{}.\note{Jer:May be worth mention the
  weight-completion-time objective.  I seem to recall seeing unrelated
  machines there.}
}

\punt{We are not aware of many scheduling problems in which jobs may run on
multiple machines at the same time.  Interestingly, Serafini,
motivated by scheduling looms in the textile industry, provides a
polynomial-time algorithm~\cite{Serafini96} for scheduling independent
jobs on unrelated machines, given that jobs can be split arbitrarily
and run independently on different machines.  This problem is quite
close to a deterministic analog of our model (with independent jobs),
but allows arbitrarily small splits (as opposed to unit-step
allocations), yielding a simpler linear-programming solution that is
not applicable to our problem.}

\subheading{Our results}
\begin{table}[t]

\centering
\begin{tabular}{c c c}
\hline\hline
Precedence Constraints & Lin and Rajaraman~\cite{LinRa07} & This work \\ [0.5ex]
\hline
Independent & $\bigo{\log(n)}$ & $\bigo{\log \log(\min\set{m,n})}$ \\
Disjoint Chains & $\bigo{\frac{\log (m) \log (n) \log(n+m)}{\log \log(n+m)}}$ & $\bigo{\log(n+m) \log\log(\min\set{m,n})}$ \\
Directed Forests & $\bigo{\frac{\log (m) \log^2 (n) \log(n+m)}{\log \log(n+m)}}$ & $\bigo{\log(n+m)\log (n) \log\log(\min\set{m,n})}$ \\ [1ex]
\hline
\end{tabular}
\caption{Improved approximation ratios}
\label{table:comparison}
\end{table}

We give improved approximation algorithms for \suu when jobs are
independent (there are no precedence constraints), and when the
precedence constraints form disjoint chains.  For independent jobs, we
give an $\bigo{\log\log(\min\set{m,n})}$-approximation algorithm.  One
component of this algorithm is based on an LP relaxation.  

Our analysis for independent jobs relies on a competitive
analysis~\cite{SleatorTa85}.  Essentially, we show that our algorithm
is $O(\log(p_{\max}/p_{\min}))$-competitive for a deterministic
scheduling problem (similar to $R|\id{pmtn}|C_{\max}$) in which each
machine has a deterministic speed, but processing times for jobs 
are chosen arbitrarily by an adversary, with 
minimum and maximum values $p_{\min}$ and $p_{\max}$.  
This competitive result is interesting in its own right.

When the precedence constraints form a collection of disjoint chains,
we have an $O(\log(n+m) \log\log(\min\set{m,n}))$-approximation
algorithm.  Our disjoint-chains algorithm uses an LP relaxation similar
to the one used for independent jobs.  We also apply techniques from
network-flow theory and prior work on the SUU problem for
chains~\cite{LinRa07}.  The $\log\log(\min\set{m,n})$ factor arises
from the independent-jobs algorithm --- therefore improving that algorithm
immediately yields a better algorithm for chains.

Our algorithm for disjoint chains can be extended to yield an
$O(\log(n+m)\log(n) \log\log(\min\set{m,n}))$-approximation for
directed forests using the chain-decomposition techniques
of~\cite{KumarMaPa05,LinRa07}.

We also show how to apply our algorithms to similar variants in the
problem of stochastic scheduling, where jobs have stochastic
processing times.  To the best of our knowledge, these are the first
approximation algorithms for stochastic scheduling with the
expected-completion-time objective and \emph{unrelated} machines.

\subheading{Paper organization} 

In \secref{prelim} we give formal definitions of the \suu problem, the
scheduling algorithms that we apply to it, and an equivalent
formulation of \suu that plays an important role in our approximation
algorithms.  A full treatment of this reformulation is given in
\appref{reform}.  \secref{indep} presents our algorithms for
independent jobs, and \secref{constr} shows how to extend them to
handle precedence constraints forming chains.  We defer tree-like
precedence constraints and stochastic scheduling to
\appreftwo{trees}{stoch}, respectively.

\secput{prelim}{Preliminaries}

In this section, we give a formal statement of our problem and then
define what we mean by a schedule.  Most of our notation is
consistent with that of Malewicz~\cite{Malewicz05} or Lin and
Rajaraman~\cite{LinRa07}. We then present a reformulation of the \suu
problem, which we use in subsequent sections to simplify both our algorithms
 and the analysis involved.

\punt{ One minor difference is that we use
``failure probabilities'' $q_{ij}$ in lieu of  ``success
probabilities'' $p_{ij}$. }

\subheading{The \suu problem}

An instance $I=(J,M,\set{q_{ij}},G)$ of the \suu problem includes
a set $J$ of unit-step jobs and a set $M$ of machines.  Throughout
this paper, we let $n = \card{J}$ be the number of jobs and
$m=\card{M}$ be the number of machines.  For each machine $i$ and job
$j$, we are given a \defn{failure probability} $q_{ij}$, which is the
probability that job $j$ \emph{does not} complete when run on machine
$i$ for one unit step; these probabilities are independent.  Without
loss of generality, we assume that for each job $j$, there exists a
machine $i$ such that $q_{ij} < 1$.

An \suu instance also includes a set of precedence constraints
comprising a directed acyclic graph (dag)~$G$ with jobs as vertices.
We say that a job $j$ is \defn{eligible} for execution at time $t$ if
all jobs preceding $j$ (i.e, jobs having a directed path to $j$) in
the dag have successfully completed before time~$t$.  If a job~$j$ is
eligible at time~$t$, a schedule may assign multiple machines
$M_{j,t} \subseteq M$ to execute $j$ in parallel.  As all
machine/job failures are independent, the probability that $j$
\emph{does not} complete in that timestep is $\prod_{i \in M_{j,t}}
q_{ij}$.

Failure probabilities are difficult to work with because they
multiply.  Instead, we define the \defn{log failure} of job $j$ on
machine $i$, denoted by $\ell_{ij}$, as $\ell_{ij} = -\log{q_{ij}}$.
Here and throughout the paper, we use $\log$ to mean a base-$2$
$\log$.  Note that by definition, $q_{ij} = 1/2^{\ell_{ij}}$, and
hence \vspace{-.5em} $\prod_{i \in M_{j,t}} q_{ij} = 1/2^{\sum_{i
    \in M_{j,t}} \ell_{ij}}$.  We define the \defn{log mass} of an
assignment to job $j$ in step~$t$ be the sum of the log failures,
given by $\sum_{i\in M_{j,t}} \ell_{ij}$.  We also use log mass to
refer to the sum of log masses across multiple timesteps; that is, the
log mass accrued from step $t_1$ through step $t_2$ is given by $\sum_{t =
  t_1}^{t_2} \sum_{i\in M_{j,t}} \ell_{ij}$.

Our work focuses on finding scheduling algorithms that minimize
expected makespans for restricted classes of precedence constraints.
If there are no precedence constraints, we say that the jobs are
\defn{independent}, and refer to the problem as \defn{\suui}.  When
the precedence constraints form a collection of chains, we call the
problem \defn{\suuc}.  When the constraints form a collection of
disjoint trees, we call the problem \defn{\suut}.  We use sans serif
fonts to refer to problem variants, whereas serif fonts refer to
algorithms/schedules for the problem.

\subheading{Schedules}

A \defn{schedule} $\Sigma$ is a policy for assigning machines to
(uncompleted) jobs.  Jobs must be scheduled at a unit granularity, but
the schedule may assign multiple machines to the same job.  A schedule
may base its decisions on any of its history, but we concern ourselves
with only schedules that can be computed in polynomial time.  More
formally, a schedule is a function $\Sigma : (H\times \naturals)
\rightarrow (M \rightarrow J \cup \set{\bot})$ that, given a
history\footnote{A full history for a deterministic schedule can be
  captured by the sets of remaining jobs at each timestep prior to the
  current timestep~$t$.  More formally, let $H_t =
  \set{\ang{S_1,S_2,\ldots,S_t}| J=S_1 \supseteq S_2 \supseteq \cdots
    \supseteq S_t}$ denote the set of all feasible ordered sets of
  remaining jobs at timesteps $1,2,\ldots,t$.  Then valid histories
  are given by the set $H = \bigcup_{t=1}^\infty H_t$.  Note that
  compact representations of the history exist, so a polynomially
  computable schedule may consider the entire history.} $h \in H$ and
time $t\in\naturals$, returns a function assigning machines to jobs.
\punt{ \footnote{Previous definitions~\cite{Malewicz05,LinRa07} define
    the simpler $\Sigma:(S\subseteq J) \rightarrow (M\rightarrow
    J\cup\set{\bot})$, but this definition excludes reasonable
    scheduling policies.}  } We use the symbol $\bot$ to indicate that
the machine remains idle.  To allow for more concise schedules, the
assignment function returned by $\Sigma(h,t)$ may map a machine to a
job that has already completed.

We define an \defn{execution} of $\Sigma$ as follows.  Suppose that
$h\in H$ is the history of the execution up to time~$t$.  Then $\Sigma$
assigns machine $i$ to job $j=\Sigma(h,t)(i)$ at step~$t$.  If $j$ has
been completed when it is scheduled to run, $i$ is assigned to $\bot$.
Since $J$, $M$, $\set{q_{ij}}$, and $G$ are invariant over a problem
instance, we allow $\Sigma$ to reference those implicitly.

Whenever the schedule $\Sigma$ is such that it assigns machines to
jobs depending only on the current time and the initial set of jobs,
not the jobs that have completed (i.e., for all~$t$, $\Sigma(h,t) =
\Sigma(h',t)$ for all $h,h'\in H$), we say that the schedule is
\defn{oblivious}.  An oblivious schedule has finite length if it is
only defined for $t\leq t_o$, for some $t_o$.

We say that a schedule is \defn{semioblivious} if it can be decomposed
into ``rounds'' such that the assignments within each round are
characterized by finite oblivious schedules.  Thus, while executing a
step contained in a particular round, the assignment of machines to
jobs depends only on the initial set of jobs when the round began and
the number of steps the round has been running.

\punt{
Note that all
oblivious schedules are semioblivious, whereas the converse is not
necessarily true.  Our schedule for \suui is semioblivious.
}

\punt{
Malewicz~\cite{Malewicz05} and Lin and Rajaraman~\cite{LinRa07} also
define \defn{regimens} as schedules having assignment of machines to
jobs dependent only on the subset of jobs remaining.  Although
regimens appear to be the most intuitive form of schedules for the
\suu problem, none of the schedules given in this paper are regimens.
}

We let $\Tsig$ be a random variable denoting the length of the
execution of schedule $\Sigma$, which is the number of steps before
all jobs have completed.  Our objective is to minimize $\expect{\Tsig}$
(denoted by $\expect{C_{\max}}$ in much of the scheduling literature).  We
refer to a schedule that has minimum expected makespan as $\Sigma_{\rm
  OPT}$, and its expected makespan, which is finite~\cite{Malewicz05},
as $\expect{\Topt}$.  For any \suu instance, $\Sigma_{\rm OPT}$ exists, and
can be computed (inefficiently) by selecting the assignment of jobs to
machines on a particular timestep that minimizes the expected makespan
of the remaining jobs.

In this paper, we consider schedules that are polynomial-time
computable (in $n$, $m$, and $\log \expect{\Topt}$) and whose expected
makespans approximate~$\expect{\Topt}$.  We say that $\Sigma$ is an
$\alpha$-approximation if $\expect{\Tsig} \leq \alpha \expect{\Topt}$ for all
choices of probabilities $\set{q_{ij}}$.

Throughout the remainder of this paper, we use algorithm and schedule
interchangeably.  Moreover, we generally do not give the schedule
explicitly as a function assigning machines to jobs. Instead, we
describe it algorthmically.

\subheading{Problem reformulation} 

We now describe a new, and equivalent formulation of the \suu problem,
which we refer to as \suureform.  Because of their equivalence, we
refer to both problems as \suu later in the paper.  A full treatment
is presented in \appref{reform}.

An \suureform instance $I=(J,M,\set{q_{ij}},G)$ has the same structure
as an \suu instance. The difference is that rather than considering
the success or failure of a job as it runs on machines in each
timestep, we use the Principle of Deferred Decisions~\cite{Motwani95} to view the
problem as one of deterministically scheduling jobs with randomly
distributed lengths.

Instead of failure probability, in \suureform we view $\ell_{ij} =
-\log q_{ij}$ as an amount of ``work'' that a machine does towards a
job completion in each unit timestep.  As in \suu, machines must be
scheduled at a unit granularity.  At the start of a schedule's
execution, we draw for each job $j$ a single random variable $r_j$
chosen uniformly at random from the $(0,1)$ interval.  A job $j$
completes when the total work done (or,``log mass'' accrued) on $j$
exceeds $-\log r_j$.  That is, $j$ completes at the first step~$t$ in
which $\sum_{k=1}^t \sum_{i \in M_{j,k}} \ell_{ij} \geq -\log r_j$.
As schedules are oblivious to these $r_j$, they behave the same way on
\suu and \suureform instances.

\secput{indep}{Independent Jobs}

This section describes an $O(\log\log{n})$-approximation algorithm for
\suui, the \suu problem with independent jobs.  We first give an
oblivious $O(\log{n})$-approximation algorithm for \suui, based on
scheduling an (approximation of an) integer linear program.  We then
modify this algorithm into a semioblivious solution consisting of
$O(\log\log n)$ nearly optimal phases.  

\subheading{An oblivious $O(\log{n})$-approximation}

We now describe an $O(\log{n})$-approximation for \suui, called
\ALGobl. 
\punt{which represents an improvement over the oblivious
$O(\log{n}\log(\min\set{m,n}))$-approximation given in~\cite{LinRa07}.}
Our approach constructs a schedule of length
$O(\expect{\Topt})$, based on an integer linear program, such that each job
has no more than a constant probability of failure upon completion.
This finite oblivious schedule is repeated until all jobs have completed. 
Using Chernoff bounds, we conclude that the expected number of repetitions is 
$O(\log n)$,  yielding an $O(\log n)$-approximation.

We use the following integer linear program for \ALGobl.  Let $x_{ij}$
denote the number of steps during which machine $i$ is assigned to
job~$j$.  Recall $\ell_{ij} = -\log q_{ij}$ is the log failure of
job~$j$ on machine~$i$.  Let $L$ be a fixed positive real,
representing a target log mass for each job, and let $J' \subseteq J$
be a subset of jobs that need to achieve that log mass.  For now,
think of $L$ as being fixed at $L=1/2$, and $J'=J$.  We assign
different values to $L$ later for the semioblivious $O(\log\log
n)$-approximation.
\begin{eqnarray}
  \lplabel{lp1}\hspace{1cm} \min t \nonumber \\ 
  \mbox{s.t. $\sum_{i\in M} \ell_{ij}x_{ij}$} & \geq & L \;\;\; \forall j \in J' \label{eq:lp1aggfailure}\\
  \mbox{$\sum_{j\in J} x_{ij}$} & \leq & t \;\;\; \forall i \in M
  \label{eq:lp1load} \\
  x_{ij} & \in & \naturals\cup\set{0} \;\;\; \forall i\in M, j\in J \label{eq:lp1int}\ .
\end{eqnarray}
Here \eqref{lp1aggfailure} enforces that every
job in $J$ has a failure probability no greater than
$1/\sqrt{2}$, and 
\punt{
\footnote{Any positive constant works here.  A constant
  in the constraint strictly less than $1$ (corresponding to a failure
  probability strictly greater than $1/2$) simplifies the proof of
  \lemref{lplower}.}
  }
\eqref{lp1int} guarantees that all jobs
are scheduled for an integral number of steps on each machine.
We use \lpref{lp1} to refer to this integer linear program
generically, and $\lpfunc{lp1}(J',L)$ to refer to it with particular 
values of $J'$ and $L$.  We denote the optimal value for $\lpfunc{lp1}(J',L)$
by $t_{\lpfunc{lp1}(J',L)}$.

A solution for $\lpfunc{lp1}(J',L)$ naturally generalizes to a finite
oblivious schedule, denoted by $\Sigma_{\lpfunc{lp1}(J',L)}$, with
length $t_{\lpfunc{lp1}(J',L)}$ as follows.  Consider a machine~$i$,
and consider each job~$j$ in arbitrary order.  Assign machine~$i$ to
job~$j$ for $x_{ij}$ timesteps. To finish our description of this schedule, 
we first claim that $t_{\lpfunc{lp1}(J,1/2)}$ approximates $\expect{\Topt}$ 
(the proof appears in \appref{app}).
Then we show how to approximate \lpref{lp1} in polynomial time.

\setthmcount{lplower}
\newcommand{\lplowerthm}{
\begin{lemma}\lemlabel{lplower}
  $t_{\lpfunc{lp1}(J,1/2)} = O(\expect{\Topt})$
\end{lemma}}
\lplowerthm
\newcommand{\lplowerproof}{
\begin{proof}
  Let $t$ be the optimum solution to $\lpfunc{lp1}(J,1/2)$.  Consider
  any subset $U \subseteq J$, and its complement $\overline U$.  Then
  $\lpfunc{lp1}(U,1/2) + \lpfunc{lp1}(\overline U,1/2) \ge t$,
  since we can construct a solution to $\lpfunc{lp1}(J,1/2)$ by adding
  a solution to $\lpfunc{lp1}(U,1/2)$ and a solution to
  $\lpfunc{lp1}(\overline U,1/2)$.

  Now recall our view of the problem in terms of \suureform: there is
  an $r_j$ chosen uniformly at random from $(0,1)$ for each job $j$
  such that job $j$ completes only if $\sum_{i \in
  M} \ell_{ij}x_{ij} \geq -\log r_j$.  For any sample from the event
  space, let $U$ be the set of jobs $j$ for which $r_j>1/2$, and let
  $\overline U$ be the complement set of jobs $j$ for which $r_j <
  1/2$ (note $r_j=1/2$ with probability 0 so can be ignored).  By
  definition, each job is in $U$ independently with probability $1/2$.
  Next observe that however OPT generates its schedule, it must
  allocate at least $1/2$ unit of ``work'' to each job in $U$; in
  other words, \eqref{lp1aggfailure} of \lpref{lp1} must hold for
  every $j\in U$.  Thus, the optimum schedule contains a feasible
  solution to $\lpfunc{lp1}(U,1/2)$.

  Now observe that by construction, $U$ is a uniformly random subset
  of $J$, meaning all subsets are equally likely.  Thus, 
\begin{eqnarray*}
  E[T_{\rm OPT}] &= &2^{-n} \sum_U E[T_{\rm OPT} \mid U]\\
  &= &2^{-n}\cdot \frac12 ( \sum_U E[T_{\rm OPT} \mid U] + \sum_U
  E[T_{\rm OPT} \mid \overline U])\\
  & \ge &2^{-n}\frac12 \sum_U (\lpfunc{lp1}(U,1/2) + \lpfunc{lp1}(\overline
  U,1/2))\\
  & \ge & 2^{-n}\frac12 \sum_U \lpfunc{lp1}(J,1/2)\\
  &= &\frac12 \lpfunc{lp1}(J,1/2)
\end{eqnarray*}
Where the second line of this derivation follows from the first
paragraph of this proof.
\punt{
  Viewing the problem in terms of \suureform, there is an $r_j$ chosen
  uniformly at random from $(0,1)$ for each job $j$ such that job $j$
  completes only if $\sum_{i \in M} \ell_{ij}x_{ij} \geq -\log r_j$.  

  To prove the claim, we first show that with constant probability, at
  least half the jobs require a constant log mass.  We then show that
  with constant probability, a random selection of $n/2$ jobs yields
  an \lpref{lp1} instance with optimal solution
  $\Omega(T_{\lpfunc{lp1}(J,1/2)})$.  We combine these two steps to
  complete the proof.
  
  Let $C_j$ be a Bernoulli random variable with value $1$ if $1/2 \geq
  -\log r_j$, or equivalently $r_j \geq 1/\sqrt{2}$.  Said
  differently, $C_j=1$ when an execution of $\lpfunc{lp1}(J,1/2)$
  would complete job~$j$.  When $C_j=0$, any execution (and in
  particular an execution of $\Sigma_{\rm OPT}$) must give job $j$ a
  log mass exceeding $1/2$.  Our goal is thus to show that a large
  fraction of jobs have $C_j=0$.  Let $C = \sum_j C_j$.  Then $\expect{C} =
  n(1-1/\sqrt{2})$.  Applying Markov's inequality gives that $\prob{C
    \geq n/2} \leq n(1-1/\sqrt{2})/(n/2) = 2-\sqrt{2} < 1$.  Thus,
  with constant probability at least $c=1-(2-\sqrt{2})$, any schedule
  (and in particular $\Sigma_{\rm OPT}$) must assign machines to jobs
  such that at least half of them have log mass exceeding~$1/2$ (and
  hence satisfy \eqref{lp1aggfailure}).
  
  Let $t^* = t_{\lpfunc{lp1}(J,1/2)}$, and consider any uniformly
  random selection $J_r\subset J$ of $n/2$ jobs.\footnote{For
    simplicity here, without loss of generality, we assume $n$ is a
    multiple of $2$.}  Then either $t_{\lpfunc{lp1}(J_r,1/2)} > t^*/2$
  or $t_{\lpfunc{lp1}(J-J_r,1/2)} > t^*/2$, since solving the linear
  program for $J_r$ then $J-J_r$ solves it for $J$.  We observe that
  since $\card{J_r} = \card{J-J_r} = n/2$, both $J_r$ and $J-J_r$ are
  equalLy likely to be selected as the random subset.  Thus, with
  probability at least $1/2$, we have $t_{\lpfunc{lp1}(J_r,1/2)} > t^*/2$
  
  We thus have that with probability at least $c/2$, $\Sigma_{\rm
    OPT}$ must schedule the jobs such that at least half of them
  satisfy \eqref{lp1aggfailure}, and such that \lpref{lp1} applied to
  those jobs has length at least $t^*/2$.  We therefore conclude that
  $\expect{\Topt} \geq (c/2) t^*/2$, and hence $t_{\lpref{lp1}}(J,1/2) =
  O(\expect{\Topt})$.
}
\end{proof}
}

The following lemma states that, in polynomial time, we can find an integral
assignment that approximates \lpref{lp1} to within a constant factor.
Some aspects of the proof are similar to~\cite[Theorem 4.1]{LinRa07},
\punt{which solves a slightly different problem useful for precedence
constraints that form disjoint chains,} but we add several steps \punt{(made
possible by our formulation involving log failure)} that improve the
approximation ratio.  
\punt{
Our tighter approximation \emph{does} apply to
the more general disjoint-chains variant, which we revisit in
\secref{constr}.
}

\begin{lemma}\lemlabel{lp1rounding}
  There exists a polynomial-time algorithm that computes a feasible
  solution to $\lpfunc{lp1}(J',L)$ having value
  $O(t_{\lpfunc{lp1}(J',L)})$.
\end{lemma}
\begin{proof}
  We relax our integer linear program to a linear program, and then
  show that the relaxed LP can be rounded to yield an integral
  $\set{\widehat{x_{ij}}}$ solution with value
  $O(T_{\lpfunc{lp1}(J',L)})$.
  
  First, let $\ell'_{ij} = \min\set{\ell_{ij},L}$.  Then we replace
  each $\ell_{ij}$ in \eqref{lp1aggfailure} with $\ell'_{ij}$,
  yielding the constraint $\sum_{i\in M} \ell'_{ij} x_{ij} \geq L,
  \forall j\in J'$.  Note that since assignments are restricted to be
  integral, this change has no effect on either the feasibility or the
  value of an assignment.  Next we remove \eqref{lp1int} and solve the
  relaxed linear program.\punt{ \footnote{As far as the relaxed
      linear program is concerned, the change from $\ell_{ij}$ to
      $\ell'_{ij}$ may have a big effect on the feasibility of
      solutions.  In particular, suppose there exists a machine~$i$
      such that $\ell_{ij}=Ln$ for all $j$.  Then an LP based on
      $\ell_{ij}$ may assign all jobs to machine~$i$ for $1/n$ time
      steps.}  }  Letting $\set{x^*_{ij}, t^*}$ be an optimal
  solution, we note that $t^* \leq t_{\lpfunc{lp1}(J',L)}$, because
  integral solutions are feasible.
  
  Our goal now is to round the LP solution to an integral solution,
  while not increasing its value by very much.  We 
  proceed in three steps.  First, we group machines having
  similar $\ell'_{ij}$ for a job~$j$, yielding a single assignment for the
  whole group.  Then, we round those assignments to integers.
  Finally, we show that the rounded assignments satisfy
   \lpref{lp1}, using an integral flow network.
  
   For each job $j$, we group machines having $\ell'_{ij}$ values
   within a factor of 2, and determine the total assignment to that
   group.  More formally, for each $j$ and integer $k$, we let
   $D^*_{jk} = \sum_{i:\floor{\log \ell'_{ij}} = k}
   x^*_{ij}$\vspace{-.2em} be the total assignment of machines with
   $\ell'_{ij} \in [2^k,2^{k+1})$ to job~$j$.\punt{ \footnote{Note
       that there are at most $mn$ values of $k$ for which $D^*_{jk}
       \neq 0$.}  }  It should be clear that $\sum_{i \in M}
   \ell'_{ij} x^*_{ij} \geq \sum_k D^*_{jk} 2^k \geq L/2$, for all
   $j\in J'$.
  
  We next round the value of $D^*_{jk}$ up to $\floor{6D^*_{jk}}$.  We
  claim that $\sum_k \floor{6D^*_{jk}} 2^k \geq L, \forall j\in J'$.
  Note that since $\ell'_{ij} \not> L$, the maximum value of $k$
  having nonzero $D^*_{jk}$ is $\floor{\log L}$.  Thus, the claim
  follows because $\sum_k \floor{6D^*_{jk}} 2^k \geq 3 (2\sum_k
  D^*_{jk} 2^k) - (\sum_{k\leq \log{L}} 2^k) \geq 3 (L)
  - (\sum_{k=0}^{\infty} L / 2^k) \geq 3L - 2L = L$.  In other words,
  rounding the group assignments down to integers can only cause us to
  lose a log mass of at most $2L$.  We thus need an assignment giving
  $3L$ log mass to the job.  
  \punt{The constant $6$ appears because we lose
  a factor of $2$ in our groupings.}
  
  To complete the integral assignment, we construct a network-flow
  instance as follows.  For each job~$j$ and integer $k$, we have a
  node $u_{jk}$.  For each machine~$i$, we have a node $v_i$.  We also
  add a source-node~$s$ and a sink-node~$w$.  For each $u_{jk}$, we
  add a directed edge $(s,u_{jk})$ with capacity $\floor{6D^*_{jk}}$.
  For each $v_i$, we add a directed edge $(v_i,w)$ with capacity
  $\ceil{6t^*}$.  Finally, we add a directed edge $(u_{jk},v_i)$ with
  infinite capacity, for any $j,k,i$ such that $\floor{\log
    \ell'_{ij}} = k$.  Note that for a given~$j$ and~$i$, there is
  exactly one $k$ such that $(u_{jk},v_i)$ exists.  We refer to this
  edge as edge $(j,i)$.
   
   Note that if we make the capacity of edges $(s,u_{jk})$ be
  $6D^*_{jk}$ instead, then a flow of
  demand $\sum_{jk} 6D^*_{jk}$ exists in this network.
  \punt{
  , and it can be constructed by
  setting the flow along edge $(j,i)$ to be $6x^*_{ij}$ (and flow
  along edge $(v_i,w)$ to be $6\sum_{j\in J} x^*_{ij}$).
  }  
  Thus, a flow of capacity $\sum_{jk} \floor{6D^*_{jk}}$ exists 
  when we lower the capacity of edge $(s,u_{jk})$ to $\floor{6D^*_{jk}}$.
  
  Ford-Fulkerson's theorem~\cite{FordFu62,CormenLeRi01} states that an
  integral max flow exists whenever the capacities are integral, as
  they are here.  We therefore take the flow across the edges $(j,i)$
  as our integral assignments $\widehat{x_{ij}}$.  Moreover, by
  construction, $\widehat{x_{ij}}$ satisfy $\sum_{j\in J}
  \widehat{x_{ij}} \leq \ceil{6t^*}$ for all $i\in M$, and $\sum_{i
    \in M} \ell'_{ij}\widehat{x_{ij}} \geq \sum_k \floor{6D^*_{jk}}
  2^k \geq L$ for all $j \in J$.  We thus have an integral feasible
  solution $\set{\widehat{x_{ij}},6t^*}$.  Noting that $6t^* \leq
  6T_{\lpfunc{lp1}(J',L)}$ completes the proof.
\end{proof}

Recall that \lemref{lplower} shows that $t_{\lpfunc{lp1}(J,1/2)} =
O(\expect{\Topt})$.  Then \lemreftwo{lplower}{lp1rounding} in concert state
that in polynomial time, we can find a schedule $\Sigma$ of length
$O(\expect{ \Topt })$, such that every job has at most a constant
probability of failure. Repeating $\Sigma$ until all jobs complete
gives our oblivious schedule \ALGobl.  The proof appears in \appref{app}.
\setthmcount{indbound}
\newcommand{\indboundthm}{
\begin{theorem}\thmlabel{indbound}
  Let $T_{\ALGobl}$ denote the random variable corresponding to the
  amount of time it takes for an execution of \ALGobl to complete all
  jobs.  Then $\expect{T_{\ALGobl}} = O(\expect{\Topt} \log n)$.
\end{theorem}
}
\indboundthm
\newcommand{\indboundproof}{
\begin{proof}
  From \lemreftwo{lplower}{lp1rounding}, we have a schedule $\Sigma$
  of length $O(\expect{ \Topt } )$ that gives each job a constant
  probability of success.  Applying a Chernoff bound gives us that a
  particular job completes in $O(\log n)$ repetitions of $\Sigma$,
  with probability at least $1-1/n^{\Theta(1)}$, where the constant
  exponent appears as a constant factor in the number of repetitions.
  Taking a union bound over all jobs gives that with probability at
  least $1-1/n^{\Theta(1)}$, \emph{all} jobs complete in $O(\log n)$
  repetitions.  Since this probability drops off dramatically as the
  number of repetitions increases, we have $\expect{T_{\ALGobl}} =
  O(\expect{\Topt}\log n)$.
\end{proof}
}\vspace{-.5em}
\subheading{A semioblivious $O(\log\log(\min\set{m,n}))$-approximation}

We construct our semioblivious schedule \ALGsem as follows.  The
schedule is divided into ``rounds.''  The first round corresponds to
an execution of the schedule suggested by the (rounded) solution
to $\lpfunc{lp1}(J,1/2)$.  In each following round, \lpref{lp1} is
applied to all remaining jobs with doubling targets.  If  
$J_k \subseteq J_{k-1}\subseteq J$ are the set of jobs left at
the start of round $k$,  then for that round we find an approximate solution 
to $\lpfunc{lp1}(J_k,2^{k-2})$, and schedule obliviously according to it.

\ALGsem runs at most $K = \ceil{\log\log \min \set{m,n}}+3$ of these
rounds. If uncompleted jobs remain after the $K$th round,
one of two things is done.  If $n \leq m$, \ALGsem runs each job one at a time
on all machines, until all jobs are completed. If $m < n$, \ALGsem
simply repeats the schedule $\Sigma_{\lpfunc{lp1}(J_k,2^{k-2})}$ given
by the $K$th round until all jobs complete.

The following theorem states that \ALGsem achieves an $O(\log\log(\min
\set{m,n}))$-approximation.  The key aspect of our analysis is viewing
\ALGsem as an ``online algorithm'' to solve the \suureform problem
over the hidden input $\set{r_j}$.
\punt{
Essentially, \ALGsem ``discovers'' the values of $r_j$ as jobs
complete.
}
  We compare the length of \ALGsem's schedule against that of
an optimal offline algorithm, called OFF, that knows the values
$\set{r_j}$.  In particular, we show that if OFF takes total time $t$
on input $\set{r_j}$, then each round of \ALGsem takes time $O(t)$ on
the same input.  This part of our proof is essentially a competitive
analysis~\cite{SleatorTa85}.

\setthmcount{indepm}
\newcommand{\indepmthm}{
\begin{theorem} \thmlabel{indepalg}
  Let $K = \ceil{\log\log \min \set{m,n}}+3$ and let
  $T_{\ALGsem}$ denote the random variable corresponding to
  the amount of time it takes for an execution of \ALGsem to complete
  all jobs. Then $\expect{T_{\ALGsem}} = O(\expect{\Topt}\cdot K)$. 
\end{theorem}
}
\indepmthm
\begin{proof}
  First, we show that for any fixed set of random values $\set{r_j}$,
  each round of \ALGsem takes time proportional to the optimal
  strategy~OFF.  Then, we show that the expected time to complete 
  the first $K$ rounds is $O(\expect{\Topt} \cdot K)$.  Finally, we analyze the
  cases when \ALGsem does not complete in $K$ rounds.
  
  Consider an optimal offline strategy OFF that knows the random
  values of $\set{r_j}$, and let $\Toff(\set{r_j})$ denote OFF's
  makespan given the values $\set{r_j}$ (i.e., $\Toff(\set{r_j})$ is
  the minimum over all strategies for a fixed $\set{r_j}$).
  For each $k\in\set{2,3,\ldots,K}$, let $J_k \subseteq J$ be the
  subset of jobs such that $-\log r_j > 2^{k-3}$.  Thus, OFF must
  assign machines to job~$j\in J_k$ such that
  $\sum_{i\in M} x_{ij} \ell_{ij} \geq 2^{k-3}$. Thus,  we have
  $\Toff(\set{r_j}) \geq T_{\lpfunc{lp1}(J_k,2^{k-3})}$.

  Now consider an execution of \ALGsem for the same $\set{r_j}$.  We
  note that if a job $j$ remains uncompleted at the start of the $k$th
  round, for $k\in\set{2,3,\ldots K}$, then $-\log r_j > 2^{k-3}$.
  This inequality follows from the fact that in the $(k-1)$th round,
  we give every job a log mass that exceeds $2^{k-3}$.  Hence,
  the jobs executed in the $k$th round are a subset of those
  defined by $J_k$ above.  By \lemref{lp1rounding}, the $k$th
  round takes time $O(T_{\lpfunc{lp1}(J_k,2^{k-2})})$.  Observing that
  $T_{\lpfunc{lp1}(J_k,2^{k-2})} \leq 2T_{\lpfunc{lp1}(J_k,2^{k-3})}$,
  we conclude that \ALGsem's $k$th round takes time
  $O(\Toff(\set{r_j}))$.
  
  We thus have that for any particular $\set{r_j}$, if \ALGsem
  completes in $d+1\leq K$ rounds, then it takes total time
  $O(\expect{\Topt} + d\Toff(\set{r_j}))$.  The $\expect{\Topt}$ term
  comes from the time it takes to execute the first round (as
  in \lemref{lplower}).  Since $\Topt \geq \expect{\Toff}$, where $\Toff$ is
  the time taken by OFF on a randomly selected $\set{r_j}$, we
  conclude that the expected time for \ALGsem's first $d+1 \leq K$
  rounds is $O(\expect{\Topt}\cdot d) = O(\expect{\Topt}\cdot K)$.
  
  We now consider the case when \ALGsem has not completed after $K$
  rounds.  Let $F_K$ be Bernoulli r.v. that is 1 if all jobs have
  completed after the $K$th round and 0 otherwise. Clearly,
  $\expect{T_{\ALGsem} | F_K = 1} = O(\expect{\Topt} \cdot K)$.  To
  complete the proof, we must show that $\prob{F_K = 0} \expect{T_{\ALGsem}
  | F_K = 0 } \leq O(\expect{\Topt} \cdot K)$.  We consider the remainder of
  the proof in two cases, depending on whether $n \leq m$ or $m < n$.
  The case of $m<n$ appears in \appref{app}.
  
  Suppose $n \leq m$ and $F_K = 0$.  Recall that after the $K$th
  round, \ALGsem runs jobs one after the other.  Recall also that
  remaining jobs must have $- \log r_j \geq 2^{K-2} \geq 2^{\log\log n +
  1} \geq 2\log n$.  This event occurs
  when $r_j < 1/n^2$, which, by the union bound, happens with
  probability no more than $1/n$.  Running jobs one at a time is trivially
   an $O(n)$-approximation, so we conclude
  that $\prob{F_K = 0}\expect{T_{\ALGsem} | F_K = 0 \land n \leq m}
  \leq O(\expect{\Topt} K)$.
\end{proof}  

\newcommand{\indepmproof}{
\begin{proof}
  This proof resumes where the same proof from \secref{indep} leaves off.
  In particular, we show here that \ALGsem produces an $O(\log\log
  m)$-approximation by showing that repeating the $(\ceil{\log\log m}
  + 3)$rd round until completion takes time $O(\expect{\Topt})$ .
  
  Suppose that $m < n$ and $F_K = 0$.  Let $\Sigma_K =
  \Sigma_{\lpfunc{lp1}(J_K,2^{k-2})}$ be the schedule computed at the
  end of the $K$th phase. Here, \ALGsem repeats $\Sigma_K$ until all
  jobs complete. Define the load $H$ of a finite schedule to be the
  maximum number of timesteps during which any machine is assigned to
  an uncompleted job (i.e., $H = \max_i \sum_j x_{ij}$); a schedule
  can be compressed to run in exactly $H$ time-steps.  We will analyze
  how long it takes for the load of our instance to drop from its
  initial expected value $O(\expect{\Topt})$ (at the end of the $K$th
  phase) to 0 (when all jobs have completed).
  
  Define $X$ as a random variable denoting the number of timesteps
  until the compressed load of $\Sigma_K$ drops to 0, and let $T_K$ be
  the (random variable) denoting the length of $\Sigma_K$.  Then $X =
  \sum_{i = 0}^{\log T_K} X_i T_K/2^{i} = T_K \sum_{i=0}^{\log T_K}
  X_i / 2^{i}$, where $X_i$ is the random variable representing the
  number of repetitions of $\Sigma_K$ necessary to drop its load from
  $T_K/2^{i}$ to $T_K/2^{i+1}$.  By construction, a single execution
  of $\Sigma_K$ ensures that each job remains with probability at most
  $1/m^2$, and thus the expected load of each machine shrinks by at
  least a (multiplicative) factor of $m^2$.  Markov's inequality gives
  us that each machines load decreases by a a factor of $m/2$. with
  probability at least $1/2m$.  Taking a union bound over all machines
  gives us that the load of \emph{all} machines decreases by a factor
  of $m/2$. with probability at least~$1/2$.
  
  For $m >4$, we now have that each repetition of $\Sigma_K$ decreases
  the remaining load by a factor of 2 with probability at least $1/2$,
  and hence $\expect{X_i} \leq 2$ for all~$i$---requiring only an
  expected constant number of repetitions to reduce the load by a
  constant factor.  We now have that $E\left[\sum_{i=0}^{\log T_K} X_i
    / 2^i\right] \leq \sum_{i=0}^{\infty} \expect{X_i}/2^i \leq O(1)$.
  Since each $X_i$ and $T_K$ are independent, it follows that
  $\expect{X} \leq \expect{T_K} \sum_{i=0}^{\infty} \expect{X_i}/2^i =
  O(\expect{T_K})$.  Having $\expect{T_K} \leq O(\expect{\Topt})$
  completes the proof.
\end{proof}
}

\secput{constr}{Jobs with Chain-Like Precedence Constraints}

This section gives our
$O(\log(n+m)\log\log(\min\set{m,n}))$-approximation for \suuc, the
case when precedence constraints form a collection of disjoint chains.
\punt{
In contrast, Lin and Rajaraman~\cite{LinRa07} give an $O(\log m \log
n \frac{\log(n+m)}{\log\log(n+m)})$-approximation algorithm.
}
Our
algorithm may be used as a subroutine for \suut, the more general case
where precedence constraints form disjoint trees (see
\appref{trees}). 

In \suuc, the dependency graph $G$ is a collection disjoint chains
$G=\set{C_1,C_2,\ldots,C_z}$, where each $C_k$ gives a total
order on a subset of jobs.  If job $j_1$ precedes $j_2$ in a
chain, we write $j_1 \prec j_2$.

Our algorithm for disjoint chains is similar to Lin and Rajaraman's
algorithm~\cite{LinRa07}, but we achieve a better approximation ratio
through various improvements.  We first give an overview of the
algorithm.  We provide more details later in the section.

To construct our schedule, we first find assignment $\set{x_{ij}}$ of
machines to jobs (where $x_{ij}$ is an integral number of steps for
which machine~$i$ is assigned to job~$j$), giving each job one unit of
log mass, such that the ``length'' and ``load'' of the assignment are
bounded by $O(\expect{\Topt})$.  The \defn{load} of a machine is the number
of timesteps for which any job is assigned to it (i.e, $\sum_j
x_{ij}$), and the load of the assignment is the maximum across all
machines.  The \defn{length} of a chain is the sum of the length of
the jobs in the chain.  The length of a job $j$, denoted by $d_j$, is
the maximum number of steps for which $j$ is assigned to a single
machines (i.e., $d_j = \max_i x_{ij}$).  Clearly, a schedule taking
time $T$ must have a length and load no more than~$T$.

We use an LP relaxation (similar to \lpref{lp1} in \secref{indep}) to
generate our assignment.  Details appear later in the section.  As in
\secref{indep}, our LP relaxation achieves an $O(1)$-approximation
\punt{,
representing an $O(\log m)$-factor improvement over the LP for chains
in~\cite{LinRa07}
}
.  Note that this assignment does not immediately
yield a schedule.

As we transform our assignment into an adaptive schedule, we treat long
and short jobs differently.  We say that a job is \defn{short} if the
length of its assignment is at most some value $\gamma$, to be defined
later, and the job is \defn{long} otherwise.  To simplify
presentation, suppose for now that all jobs are short.  We later
describe how to deal with long jobs.  

We then transform the assignment into an \emph{adaptive} schedule
$\Sigma_k$ for each chain $C_k$.  The schedule $\Sigma_k$ considers
the next eligible (uncompleted) job~$j$ in $C_k$, and (obliviously)
schedules the next $d_j$ timesteps according to the assignment
$\set{x_{ij}}$.  Specifically, if $\Sigma_k$ begins executing job $j$
at time~$t$, then it schedules $j$ from time $t$ to $t+x_{ij}$ on
machine $i$.  (Machine $i$ remains idle from time $t+x_{ij}$ to
$t+d_j$.)  After the $d_j$ timesteps, $\Sigma_k$ again considers the
next eligible job in the chain (which may be the same job if it
failed).  We note that each time job $j$ is obliviously scheduled, it
has a constant probability of success.  

We then combine all the $\Sigma_k$ in a straightforward manner,
yielding a ``pseudoschedule'' for the \suuc instance, denoted by
$\set{\Sigma_k}$.  In particular, a \defn{pseudoschedule} runs all
$\Sigma_k$ ``in parallel,'' possibly assigning multiple jobs to the
same machine in each timestep.  To avoid confusion, we call each of
the timesteps of a pseudoschedule a \defn{superstep}, and we call the
number of jobs assigned to a single machine during a superstep $t$ the
\defn{congestion} at that superstep, denoted by $c(t)$.  We
``flatten'' each superstep to $c(t)$ timesteps by arbitrarily ordering
the jobs assigned to each machine, thus yielding a schedule called
\ALGchn.  If $c_{\max}$ is the maximum congestion over all supersteps,
and $Z$ is the maximum length of any chain, then \ALGchn comprises
$O(c_{\max} Z)$ timesteps.  

To reduce congestion, we apply a random-delay
technique~\cite{LeightonMaRa94,ShmoysStWe94}, also used by Lin and
Rajaraman~\cite{LinRa07}.  We also utilize the fact that when chains
consist of sufficiently many (short) jobs, the number of supersteps
spanned by $\Sigma_k$ is near the expected length of $\Sigma_k$, with
high probability.
\punt{
The adaptiveness of $\Sigma_k$, not present in~\cite{LinRa07}, results
in an $O(\log n)$-factor improvement for chains consisting of short
jobs.  
}To deal with long jobs, we run \ALGsem $O(\log (n+m))$ times,
which dominates the runtime, yielding the
$O(\log(n+m)\log\log(\min\set{m,n}))$-approximation.

\subheading{Finding an assignment with low load and length.} 

As in \secref{indep}, we use an integer linear program to optimize for
the constraints.  This integer linear program for chains matches that
used in~\cite[LP1]{LinRa07}.  
\begin{eqnarray}
  \lplabel{lp2}\hspace{1cm} \min t \nonumber \\ 
  \mbox{s.t. $\sum_{i\in M} \ell_{ij}x_{ij}$} & \geq & 1 \;\;\; \forall j \in J \label{eq:lp2aggfailure}\\
  \mbox{$\sum_{j\in J} x_{ij}$} & \leq & t \;\;\; \forall i \in M
  \label{eq:lp2load}\\
  \mbox{$\sum_{j \in C_k} d_j$} & \leq & t \;\;\; \forall C_k \in G \label{eq:lp2chain}\\
  0 \leq x_{ij} & \leq & d_j \;\;\; \forall i\in M, j\in J \label{eq:lp2length}\\
  d_j & \geq & 1 \;\;\; \forall j\in J \label{eq:lp2lengthlower}\\
  x_{ij} & \in & \naturals\cup\set{0} \;\;\; \forall i\in M, j\in J \label{eq:lp2int}\ .
\end{eqnarray}
\eqrefthree{lp2aggfailure}{lp2load}{lp2int} correspond to
\eqrefthree{lp1aggfailure}{lp1load}{lp1int}, respectively, in
\lpref{lp1}.  \eqref{lp2load} bounds the load of each machine.
\eqref{lp2chain} bounds the length of each chain, and
\eqreftwo{lp2length}{lp2lengthlower} determines the length of each
job.

The following lemma, proven in~\cite[Lemma 4.2]{LinRa07}, states that
the optimal value for \lpref{lp2} is a lower bound on $\expect{\Topt}$.  

\begin{lemma}\lemlabel{lp2lower}
  Let $t_{\lpref{lp2}}$ be the optimal value for \lpref{lp2}.  Then
  $t_{\lpref{lp2}} = O(\expect{\Topt})$.  \qed
\end{lemma}

The next lemma exhibits an $O(1)$-approximation to \lpref{lp2}.
\lemreftwo{lp2lower}{lp2rounding} together imply a polynomial-time
algorithm giving an integral assignment $\set{x_{ij}}$ of machines to
jobs, such that the load and length are both $O(\expect{\Topt})$.

\begin{lemma}\lemlabel{lp2rounding}
  Let $t_{\lpref{lp2}}$ be the optimal value for \lpref{lp2}.  There
  exists a polynomial-time algorithm that computes a feasible solution
  to \lpref{lp2} having value $O(t_{\lpref{lp2}})$.
\end{lemma}
\begin{proof}
  The rounding proceeds as in \lemref{lp1rounding}, starting by
  removing \eqref{lp2int} and replacing \eqref{lp2aggfailure} by
  $\sum_{i\in M} \ell'_{ij}x_{ij} \geq 1$, for $\ell'_{ij} =
  \min\set{\ell_{ij},1}$.  The only major difference is in the
  capacity of some edges in the flow network.  Instead of giving edge
  $(j,i)$ an infinite capacity, we restrict the capacity of edge
  $(j,i)$ to $\ceil{6d_j^*}$, where $d_j^*$ is the assignment given by
  the optimal solution to the relaxed linear program.  We note that
  the length of a chain $C_k$ may increase up to at most $6\sum_{j\in
    C_k}d_j^* + \card{C_k} \leq 7\sum_{j\in C_k} d_j^*$.
\end{proof}

\subheading{Reducing congestion of \ALGchn}

As described thus far, \ALGchn may have $\Theta(n)$ congestion.  We
take advantage of a random-delay
technique~\cite{LeightonMaRa94,ShmoysStWe94} to reduce congestion to
$O(\frac{\log(n+m)}{\log\log(n+m)})$, with high probability.
Essentially, we modify \ALGchn to simply delay the start time of each
chain by a value chosen uniformly at random from $\set{0,1,\ldots,H}$,
where $H$ is the load of \ALGchn.  

The delay technique is summed up by the following theorem, proof
omitted (as similar theorems appear elsewhere). It originates
in~\cite{LeightonMaRa94}, and Lin and Rajaraman~\cite[Section
4.1]{LinRa07} outline the necessary proof as applied to \suuc.

\begin{theorem}\thmlabel{congestion}
  Consider a pseudoschedule $\set{\Sigma_k}$ with total
  load $H$, where $H$ is polynomially bounded in $n$ and
  $m$.  Consider the pseudoschedule $\set{\Sigma_{k'}}$ generated by
  randomly shifting, or ``delaying,'' the start time of each chain
  schedule $\Sigma_k$ by a value chosen uniformly at random from
  $\set{0,1,\ldots,H}$.  Then $\set{\Sigma_{k'}}$ has congestion at
  most $O(\frac{\log(n+m)}{\log\log(n+m)})$, with high probability with respect
  to $n$ and $m$.\qed
\end{theorem}
Notice that whenever the load of and length of $\set{\Sigma_k}$ are
bounded by $O(\expect{\Topt})$, it follows that the length of
$\set{\Sigma_{k'}}$ is at most $O(\expect{\Topt})$ supersteps, with high
probability.

Since $\Sigma_k$ repeats the assignment for some jobs, the load and
length of the pseudoschedule $\set{\Sigma_{k'}}$ are random variables.
We note, however, that the random successes and failures of jobs (and
hence load and length) are independent of the initial random delay
selected.  Thus, as long as our random execution yields a load and
length of $O(\expect{\Topt})$, then \thmref{congestion} implies that \ALGchn
consists of $O(\expect{\Topt})$ supersteps, for a total time of
$O(\expect{\Topt} c_{\max}) = O(\expect{\Topt}\frac{\log(n+m)}{\log\log(n+m)})$ steps.

The following lemma implies that most executions of \ALGchn result in
low load and length.  In this lemma, $y_j$ is the random variable
indicating the number of repetitions of job $j$'s assignment used to
complete $j$, and $d_j$ denotes the length of job~$j$'s assignment.
In the \suuc context, $\eta=n+m$.  The value $W$ here represents the
load or the length of the assignment.  The lemma states that whenever
a job has length (or causes load) that is logarithmically smaller than
the total, then the length of the chain or (load on a machine) is
close to the expectation, with high probability.  Union bounding over
all $O(n)$ chains (or $m$ machines) implies that the total length (and
load) of \ALGchn schedule is close to the expectation, with high
probability.  The proof appears in \appref{app}.

\setthmcount{prlemma}
\newcommand{\prlemma}{
\begin{lemma}
  Consider $y_j \in \naturals$ drawn from the geometric distribution
  $\prob{y_j = k \in \naturals} = (1/2)^k$, and let $1\leq d_j \leq
  W/\log \eta$ be a weight associated with each $y_j$ for any values
  such that $W\geq \sum_j 2d_j$ and $\log\eta\leq W$.  Then $\sum_j y_j
  d_j \leq O(cT)$ with probability at least $1-1/\eta^{c}$.
\end{lemma}
}
\prlemma
\newcommand{\prproof}{
\begin{proof}
  Round all the $d_j$ up to the next power of $2$.  Let $Z_k$ be the
  set of $j$ such that $T /(2^k \log \eta) < d_j \leq
  T/(2^{k-1}\log\eta)$, for $k\in\set{1,2,\ldots,\log(T/\log\eta)}$.
  We apply a Chernoff bound over all jobs in $Z_k$ to show that the
  weighted sum of their $y_j$ is near the expectation, with high
  probability.  In particular, let $B_k$ be the sum $b_k$ (to be
  decided later) Bernoulli random variables with probability $1/2$.
  Then clearly $\prob{B_k < \card{Z_k}} = \prob{\sum_{j \in Z_k} y_j >
    b}$.  We therefore apply Chernoff bounds to the $B_k$.  We will
  then union bound over all $B_k$ to complete the proof.  
  
  Since there are potentially many nonempty $Z_k$ sets, we choose the
  $b_k$ so as to give geometrically decreasing failure probabilities.
  In particular, we set $b_k = \alpha(c)(\card{Z_k} + \log \eta + k)$
  for some constant $\alpha(c)$, with $(1-1/\alpha(c))^2/2 = c$.  Then
  a Chernoff bound states that $\prob{B_k < \card{Z_k}} <
  e^{-(\log\eta + k)(1-1/\alpha(c))^2/2} \leq \eta^{-c} e^{-ck}$.
  Taking a union bound over all $k$ gives $\prob{{\rm any} B_k <
    \card{Z_k}} \leq \sum_{k=1}^{\log(T/\log\eta)} \eta^{-c}e^{-ck}
  \leq \eta^{-c} \sum_{k=1}^{\infty} e^{-ck} \leq \eta^{-c}$ when $c
  \geq \log_2 e$.  It follows that $\prob{{\rm any} \sum_{j\in Z_k}
    y_j > b_k} < \eta^{-c}$.  And hence with probability at least
  $1-1/\eta^c$, we have $\sum_j d_j y_j \leq \sum_k (T/(2^k \log\eta)
  \sum_{j \in Z_k} y_j) \leq \sum_k (Tb_k/(2^k\log\eta)) \leq O(c\sum_j
  d_j) + (T/\log\eta) \sum_k (\log\eta + k)/2^k = O(c\sum_j d_j) +
  O(T) = O(cT)$.
\end{proof}
}

We conclude that if jobs are short, where short jobs have length at
most $\gamma=t_{\lpfunc{lp2}}/\log(n+m)$, and if $t_{\lpfunc{lp2}}$ is
polynomial in $n$ and $m$, then \ALGchn takes time
$O(\expect{\Topt}\frac{\log(n+m)}{\log\log(n+m)})$ with high probability.  To
get this bound in expectation, we simply modify \ALGchn to run the
$O(n)$-approximation (as for \ALGsem) whenever congestion, load, or
length exceed the desired bounds, which occurs with probability at most~$1/n$.

\subheading{Handling long jobs} 

We now extend \ALGchn to handle jobs having length more than
$t_{\lpfunc{lp2}}/\log(n+m)$.  In the chain schedule $\Sigma_k$, we
replace each longer job by a ``pause'' of length
$t_{\lpfunc{lp2}}/\log(n+m)$.  Specifically, no job from the chain is
scheduled until $t_{\lpfunc{lp2}}/\log(n+m)$ supersteps later.  We
then divide our schedule \ALGchn into $O(\log(n+m))$ segments of
length $t_{\lpfunc{lp2}}/\log(n+m)$ supersteps.  Note that by
construction, there is at most one pause per chain per segment.  After
executing each segment, \ALGchn executes \ALGsem on the jobs
corresponding to the pauses starting in that segment (suspending the
rest of the chains until completion).  Once those long jobs complete,
\ALGsem continues to the next segment.

All of our previous analyses (that assume jobs are short) still hold.
In particular, we satisfy the requirement that all long relevant long
jobs complete before the short jobs are scheduled again.  Since there
are $O(\log(n+m))$ executions of \ALGsem, it follows that the total
expected time increases to
$O(\expect{\Topt}\cdot\log(n+m)\log\log(\min\set{m,n}))$, and hence we have
an $O(\log(n+m)\log\log(\min\set{m,n}))$-approximation.

\subheading{Extending to nonpolynomial $t_{\lpfunc{lp2}}$}  

We now address the requirement in \thmref{congestion} that load and
length be polynomially bounded in~$n$ and~$m$.  We make use of a trick
from~\cite[Section 3.1]{ShmoysStWe94}, also used in~\cite{LinRa07}.
Consider the chain schedule $\Sigma_k$ (having length
$O(t_{\lpfunc{lp2}})$, with high probability) before the random delay
is applied.  We round each assignment $x_{ij}$ down to the nearest
multiple of $t_{\lpfunc{lp2}}/nm$.  We thus treat the assignments as
integers in the range $\set{0,1,\ldots,O(nm)}$.  We can then apply the
random-delay technique (from \thmref{congestion}) to these rounded
assignments.  

The issue now is that the rounding may have decreased many
assignments, so we reinsert steps into the schedule.  In particular,
whenever executing job $j$, we reinsert steps (\emph{not} supersteps)
into the execution, executing only job $j$ during those steps.
Specifically, the execution of job $j$ may result in reinserting at
most an expected $2t_{\lpfunc{lp2}}/nm$ steps for each machine, and
hence $2t_{\lpfunc{lp2}}/n$ steps in total.  Summing across all $n$
jobs gives an expected $2t_{\lpfunc{lp2}}$ steps, thereby increasing
the total length of \ALGchn by $O(\expect{\Topt})$.

\begin{theorem}
  Let $T_{\ALGchn}$ denote the random variable indicating the time at
  which an execution of \ALGchn completes all jobs.  Then
  $\expect{T_{\ALGchn}} = O(\expect{\Topt} \cdot
  \log(n+m)\log\log(\min\set{m,n}))$.  \qed
\end{theorem}

\secput{conc}{Conclusion}

In this paper, we have presented improved approximation algorithms for
multiprocessor scheduling under uncertainty.  We believe that our
bounds our not tight.  In particular, we believe that a fully adaptive
schedule should be able to trim an $O(\log\log(\min\set{m,n}))$ factor
from our bounds.  It would also be interesting if a greedy heuristic
could achieve the same bounds.  Finally, we would be interested in
developing nontrivial approximations for more general precedence
constraints.  At first glance, however, it seems like any technique
for \suu and arbitrary precedence constraints may generalize to
$R|\id{pmtn},\id{prec}|C_{\max}$, which remains unsolved.

\bibliographystyle{plain}
\bibliography{suu}

\newpage
\renewcommand{\set}[1]            {\left\{ #1 \right\}}
\renewcommand{\abs}[1]            {\left| #1\right|}
\renewcommand{\card}[1]           {\left| #1\right|}
\renewcommand{\floor}[1]          {\left\lfloor #1 \right\rfloor}
\renewcommand{\ceil}[1]           {\left\lceil #1 \right\rceil}
\renewcommand{\ang}[1]            {\ifmmode{\left\langle #1 \right\rangle}
                                 \else{$\left\langle${#1}$\right\rangle$}\fi}
\renewcommand{\prob}[1]           {\Pr\left\{ #1 \right\}}
\renewcommand{\expect}[1]         {{\rm E}\left[ #1 \right]}
\renewcommand{\bigo}[1]           {O\! \left( #1 \right)}

\appendix
\secput{reform}{Problem Reformulation}

This section presents a full description of our reformulation of the
\suu problem.  We use this new formulation to simplify both our
algorithms and the analyses involved.

To disambiguate between the original statement of the \suu problem
given in \secref{prelim} and the new one described herein, in this
section we refer to the new formulation as \suureform.
Since we show that they are equivalent,
we refer to both problems as \suu later in the paper.

As with \suu, an \suureform instance includes a set $J$ of jobs, a set
$M$ of machines, and a set of precedence constraints forming a dag.
Also as before, the instance specifies a real $q_{ij} \in [0,1]$ for
each job $j$ and machine $i$.  In \suu, $q_{ij}$ specifies a failure
probability.  In \suureform, however, we do not view $q_{ij}$ as a
probability; instead we view $\ell_{ij} = -\log q_{ij}$ as an amount
of ``work'' that a machine does towards a job completion in each unit timestep.
As in \suu, machines must be scheduled at a unit granularity.

We model the stochastic nature of the problem in \suureform by
associating with each job $j$ a single random variable $r_j$ chosen
uniformly at random from the $(0,1)$ interval.  We say that a job~$j$
completes once the total work done (or log mass accrued) on~$j$
exceeds $-\log r_j$.  More formally, let $M_{j,t}$ be the set of the
machines assigned to job~$j$ by a schedule on time~$t$.  Then $j$
completes during the first step~$t$ in which $\sum_{k=1}^t \sum_{i \in
  M_{j,k}} \ell_{ij} \geq -\log r_j$, or equivalently $\prod_{k=1}^t
\prod_{i \in M_{j,k}} q_{ij} \leq r_j$.

Note that a schedule $\Sigma$ is oblivious to the random value~$r_j$.
Instead, it is only aware of whether a job completes in each timestep.
Thus, a schedule must make its decisions for assignments in step~$t$
based only on the surviving sets of jobs $S_1,S_2,\ldots,S_t$ for each
of the preceding timesteps.  Hence, the same schedule may be applied to
both \suu and \suureform.

The following theorem states that \suu and \suureform have the same
distribution over states of uncompleted jobs in each timestep.  

\begin{theorem}
  Consider executions $X$ and $X^*$ of schedule $\Sigma$ on
  \suu-instance~$I=(J,M,\set{q_{ij}},G)$ and corresponding
  \suureform-instance $I^*=(J,M,\set{q_{ij}},G)$, respectively, that
  run for $t-1$ timesteps.  We define the \defn{history} of the
  execution $X$ (and $X^*$) after $t-1$ steps, denoted by $h_t$ (and
  $h_t^*$), as a sequence of job subsets $h_t =
  \ang{S_1,S_2,\ldots,S_t}$, where $S_k\subseteq J$ is the set of
  uncomplete jobs remaining at the start of step~$k$ in the execution.
  For any state $\ang{S_1,S_2,\ldots,S_t}$, we have $\prob{h_t =
    \ang{S_1,S_2,\ldots,S_t}} = \prob{h^*_t =
    \ang{S_1,S_2,\ldots,S_t}}$.
\end{theorem}
\begin{proof}
  By induction on time $t$.  Initially, $\prob{h_1 = \ang{J}} =
  \prob{h^*_1 = \ang{J}} = 1$.

  Suppose $h_t = h^*_t = \ang{S_1,S_2,\ldots,S_t}$.  Then $\Sigma$
  makes the same decisions for assigning machines to jobs in step~$t$
  in both executions (i.e., $\Sigma(h_t,t) = \Sigma(h^*_t,t)$).  Let
  $M_{j,t}$ be the set of machines assigned to $j$ by $\Sigma(h_t,t)$.
  Let $S_{t+1}$ and $S_{t+1}^*$ denote the random variables indicating
  the subsets of jobs remaining after executing step $t$ in $X$ and
  $X^*$, respectively.  We will show that $S_{t+1}$ and $S_{t+1}^*$
  have the same distribution.

  For each job $j\in S_t$, the probability that $X$ \emph{does not}
  complete $j$ in step~$t$ of the \suu execution~$X$ is given by
  $\prob{j\in S_{t+1}|h_t}= \prod_{i\in M_{j,t}} q_{ij}$.

  We now consider the probability that $X$ \emph{does not} complete
  $j\in S_t$ in step~$t$ of the \suureform execution~$X^*$.  By
  definition, $j$ completes if $\prod_{k=1}^t \prod_{i\in M_{j,k}}
  q_{ij} \leq r_j$.  By assumption, since $j \in S_t$, we have
  $\prod_{k=1}^{t-1}\prod_{i\in M_{j,k}} q_{ij} > r_j$.  Thus,
  \vspace{-1em}
  \begin{eqnarray*}
    \prob{j \in S_{t+1}^*|h^*_t} &=& \prob{r_j < \prod_{k=1}^t
        \prod_{i\in M_{j,k}} q_{ij} \quad\vline\quad r_j <
        \prod_{k=1}^{t-1}\prod_{i\in M_{j,k}} q_{ij}} \\
    &=& \prob{r_j < \left(\prod_{i\in M_{j,t}} q_{ij}\right)
        \left(\prod_{k=1}^{t-1}\prod_{i\in M_{j,k}} q_{ij}\right)
        \quad\vline\quad r_j < \left(\prod_{k=1}^{t-1}\prod_{i\in
        M_{j,k}} q_{ij}\right)} \\
    &=& \prob{r_j < \prod_{i\in M_{j,t}} q_{ij} \quad\vline\quad
      r_j < 1}  \\
    &=& \prod_{i\in M_{j,t}} q_{ij} \ .
  \end{eqnarray*}

  Multiplying probabilities $\prob{j\in S_{t+1} | h_t}$ for each job
  $j\in S\subseteq S_t$ and $1-\prob{j\in S_{t+1} | h_t}$ for each job
  $j \in S_t-S$ yields $\prob{S_{t+1}=S|h_t} =
  \prob{S_{t+1}^*=S|h_t^*}$.  Applying the inductive hypothesis (i.e.,
  $\prob{h_t} = \prob{h_t^*}$) completes the proof.
\end{proof}

\begin{corollary}
  Let $\Tsig$ and $\Tsig^*$ be the random variables denoting the
  amount of time it takes to execute \suu and \suureform instances,
  respectively, using schedule $\Sigma$.  Then $\expect{\Tsig} =
  \expect{\Tsig^*}$. Thus, a schedule that gives an
 $\alpha$-approximation to \suureform also gives an
 $\alpha$-approximation to \suu.
\end{corollary}

\secput{trees}{Jobs with Tree-Like Precedence Constraints}

We can obtain algorithms for tree-like precedence constraints by
trivially applying techniques from~\cite{KumarMaPa05}, as done
in~\cite{LinRa07}.  We state the bound here without proof.

When precedence constraints form a directed forest, the technique
from~\cite{KumarMaPa05} decomposes the graph into $O(\log n)$ blocks, each
consisting of disjoint chains.  We then apply \ALGchn $O(\log n)$
times.

\begin{theorem}
  If precedence constraints form a directed forest, there exists a
  polynomially computable schedule with expected makespan
  $O(\expect{\Topt}\cdot \log(n)\log(n+m)\log\log(\min\set{m,n}))$. 
\end{theorem}

\punt{
When the precedence constraints form a collection of \defn{out trees},
having all edges directed away from the root, or \defn{in trees},
having all edges directed towards the root, we can apply another
technique from~\cite{KumarMaPa05} to improve the approximation.  
}

\secput{stoch}{Stochastic Scheduling} 

This section shows how our algorithms from Sections \ref{sec:indep}
and \ref{sec:constr} apply to the problem of preemptively scheduling
jobs whose lengths are given by random variables on unrelated parallel
machines. Specifically, we give polynomial time algorithms for
problems of the form $R|\id{pmtn},
\id{prec},p_j\!\!\sim\!\!\id{stoch}|\expect{C_{\max}}$ with approximation
ratios for each type of precedence constraint identical identical to
those of previous sections, so long as job lengths are drawn from
exponential distributions with known means.

First, we review the stochastic model and how it differs from \suu.
Then, we overview an $O(\log \log n)$-approximation for
$R|\id{pmtn},p_j\!\!\sim\!\!\id{stoch}|\expect{C_{\max}}$, which we
designate \stochi. Finally, we discuss briefly how to generalize to
cases with precedence constraints.

\subheading{Preliminaries}
An instance $I_{\rm stoch}=(J,M,\set{\lambda_j},\set{v_{ij}})$ of
\stochi contains a set of jobs $J$ and a set of machines $M$ just as
in \suu.  For every job~$j$, $\lambda_j$ specifies the rate parameter of the
exponential distribution from which $j$'s length, denoted by the
random variable $p_j$, is drawn. That is, $\prob{p_j \leq  c} = 1 - e^{-c\lambda_j}$.
 Only $\lambda_j$ is given in an
instance, and $p_j$ is not revealed until the job completes.  Finally for
each machine $i$ and job $j$, $v_{ij}$ specifies the speed with which
machine $i$ processes job $j$.  Specifically, let $x_{ij}$ be the
amount of time during which machine $i$ processes job $j$.  Then $j$
completes once $\sum_i x_{ij}v_{ij} \geq p_j$. This inequality should
look very similar to \suureform, which is why our earlier algorithms
apply here.

In \stochi, $x_{ij}$ need not be integral, but we do require that no
job be processed by more than one machine at the same time.  

We continue to refer to optimal algorithms by OPT and (the random
variable denoting) the time they take to run by $\Topt$.

\subheading{An $O(\log \log n)$-approximation for \stochi}

We now show how to provide a $O(\log \log n)$-approximation for
\stochi, using arguments very similar to those in \secref{indep} and a
constant-factor approximation algorithm for $R|pmtn|C_{\max}$.

Our algorithm \ALGstci operates similarly to \ALGsem. In particular,
\ALGstci operates in $K = \ceil{\log\log n +3}$ rounds, each
corresponding to an oblivious schedule $\Sigma_k$.  We construct the
oblivious $\Sigma_k$ such that any job having (stochastically chosen)
$p_j \leq 2^{k-2}/\lambda_j$ completes. Specifically, $\Sigma_k$
corresponds to (approximately) solving the \emph{deterministic} analog
$R|\id{pmtn}|C_{\max}$, setting the length of job $j$ to
$2^{k-2}/\lambda_j$.  Any jobs remaining after the end of these $K$
rounds is run one at a time on the fastest possible machine.

We use the algorithm from Lawler and Labetoulle~\cite{Lawler78} to
compute an $O(1)$-approximation for $R|\id{pmtn}|C_{\max}$ in
polynomial time, giving us each of our~$\Sigma_k$.

The following theorem states that \ALGstci approximates \stochi.
Proof (omitted) is similar to \thmref{indbound} and \lemref{lplower}

\begin{theorem}
  Let $T_{\ALGstci}$ be the random variable denoting the time it takes
  for an execution of \ALGstci to complete all jobs.  Then
  $\expect{T_{\ALGstci}} = O(\expect{\Topt})$.
\end{theorem}
\begin{proofsketch}
  The full proof includes a component similar to~\lemref{lplower},
  showing that the first round (solving for deterministic lengths
  $1/(2\lambda_j)$) approximates $\expect{\Topt}$, and a component similar
  to \thmref{indbound}, using an offline-algorithm argument to prove
  that the $K-1$ subsequent rounds take expected time $O(\expect{\Topt}\cdot
  K)$.
  
  To complete the proof, we note that when $p_j$ are bounded above by
  $2\log n/\lambda_j$, then all jobs complete during (or before)
  $\Sigma_K$.  Since the $p_j$ are exponentially distributed,
  $\prob{\exists j\;\mbox{s.t.}\;p_j > 2 \log n /\lambda_j} \leq
  1/n$. Running jobs sequentially is an $n$-approximation, but this
  sequential run occurs only with probability at most $1/n$.
\end{proofsketch}

A virtually identical algorithm gives an
$O(\log\log(n))$-approximation to the slightly weaker setting,
$R|\id{restart},p_j\!\!\sim\!\!\id{stoch}|\expect{C_{\max}}$.  In this
setting, a job must run fully on a single machine, but it may be
restarted on a different machine.  Here, job lengths are
stochastically chosen only once.  The only necessary change to the
algorithm is substitute the $k$th round with the corresponding
solution to $R||C_{\max}$, in lieu of $R|\id{pmtn}|C_{\max}$.

\subheading{Other results} 

Similar analysis yields an $O(\log\log m)$-approximation for \stochi.
Substituting assignments generated from algorithms for $R|\id{pmtn},
\id{chains} | C_{\max}$ \cite{Lenstra90} for those given by \lpref{lp2}
gives an $O( \log(n+m)\log\log(\min\set{m,n}))$-approximation when
precedence constraints form chains, and an
$O(\log(n)\log(n+m)\log\log(\min\set{m,n}))$-approximation when they
form directed forests, using the same algorithms and techniques.

\secput{app}{Proofs}

\setthmcount{oldtheorem}

\setcounter{theorem}{\thelplower}
\lplowerthm
\lplowerproof

\setcounter{theorem}{\theindbound}
\indboundthm
\indboundproof

\setcounter{theorem}{\theindepm}
\indepmthm
\indepmproof

\setcounter{theorem}{\theprlemma}
\prlemma
\prproof

\setcounter{theorem}{\theoldtheorem}

\end{document}